\documentclass[11pt]{amsart}
\usepackage{graphicx,amssymb,amsmath,amsthm}
\usepackage{enumerate}
\usepackage{dsfont}
\usepackage[colorlinks, citecolor=red]{hyperref}
\usepackage{comment,cite,color}
\usepackage{cite,color}

\usepackage{mathrsfs}

\usepackage{epsfig}
\usepackage{lscape}
\usepackage{subfigure}
\usepackage{epstopdf}
\usepackage{caption}
\usepackage[linesnumbered,ruled,vlined]{algorithm2e}

\usepackage{tgtermes}
\usepackage{lipsum}

\usepackage{bm}

\newcommand{\xkh}[1]{\left(#1\right)}
\newcommand{\dkh}[1]{\left\{#1\right\}}

\newcommand{\norm}[1]{\|{#1}\|_2}

\newcommand{\abs}[1]{\left\lvert#1\right\rvert}

\newcommand{\argmin}[1]{\mathop{\rm argmin}\limits_{#1}}

\newcommand{\T}{\top}
\newcommand{\C}{{\mathbb C}}

\newcommand{\hx}{{\widehat{\bm{x}}}}
\newcommand{\z}{{\widehat{\bm{z}}}}

\newcommand{\vx}{{\bm x}}
\newcommand{\x}{{\widehat{\bm{x}}}}

\newcommand{\vr}{{ \bm{ r}}}

\newcommand{\vz}{{\bm z}}

\newcommand{\Z}{{\mathbb Z}}

\newcommand{\wb}[1]{\overline{#1}}

\newcommand{\RNum}[1]{\uppercase\expandafter{\romannumeral #1\relax}}

\newtheorem{definition}{Definition}[section]

\newtheorem{theorem}[definition]{Theorem}
\newtheorem{lemma}[definition]{Lemma}

\newtheorem{remark}[definition]{Remark}

\date{}

\begin{document}
\baselineskip 19pt
\bibliographystyle{plain}
\title{No existence of linear algorithm for Fourier phase retrieval}
\author{Meng huang}
\address{School of Mathematical Sciences, Beihang University, Beijing, 100191, China} \email{menghuang@buaa.edu.cn}

\author{Zhiqiang Xu}
\thanks{{{Zhiqiang Xu is supported by the National Science
Fund for Distinguished Young Scholars (12025108) and NSFC (12021001)}}}
\address{LSEC, ICMSEC, Academy of Mathematics and Systems Science, Chinese Academy of Sciences, Beijing 100190, China;\newline
School of Mathematical Sciences, University of Chinese Academy of Sciences, Beijing 100049, China}
\email{xuzq@lsec.cc.ac.cn}

\maketitle
\begin{abstract}
Fourier phase retrieval, which seeks to reconstruct a signal from its Fourier magnitude, is of fundamental importance in fields of engineering and science. In this paper, we give a theoretical understanding of algorithms for Fourier phase retrieval. Particularly, we show if   there exists an algorithm which could reconstruct an arbitrary  signal $\vx\in \C^N$ in  $ \mbox{Poly}(N) \log(1/\epsilon)$ time to reach $\epsilon$-precision from its magnitude of discrete Fourier transform and its initial value $x(0)$, then $\mathcal{ P}=\mathcal{NP}$.  This demystifies the phenomenon that, although almost all signals are determined uniquely by their Fourier magnitude with a prior conditions,  there is no algorithm with theoretical guarantees being proposed over the past few decades.
%In particular, our result shows that the natural SDP relaxation such as PhaseLift is not suitable to solve the Fourier phase retrieval problem.
Our proofs employ the  result in computational complexity theory that Product Partition problem is NP-complete in the strong sense.
\end{abstract}

\section{Introduction}

\subsection{Problem setup}
Let $\vx:=(x(0),\ldots, x(N-1))^\T \in \C^N$ be the underlying signal.  The Fourier phase retrieval problem aims to reconstruct the signal $\vx$ from the magnitude of its discrete Fourier transform
\begin{equation} \label{eq:probset}
|\hx(\omega)|:=\abs{\sum_{n=0}^{N-1} x(n) e^{-i\omega n}}, \quad \omega \in [0,2\pi).
\end{equation}
This problem can also be understood as the reconstruction of  $\vx$ from the autocorrelation signal $\vr=(r(n))_{n\in \Z}$ where
\begin{equation*}
r(n):=\sum_{k=0}^{N-1} \overline{x(k)}  x(k+n), \quad n=-N+1,\ldots,N+1
\end{equation*}
with $x(k)=0$ for $k<0$ and $k\ge N$. Here, $\overline{x(k)}$ is the conjugate  of $x(k)$. The Fourier phase retrieval problem has a rich history tracing back to 1952 \cite{sayre1952some}.
%and has challenged engineers, physicists and mathematicians for decades.
It  has a broad range of applications  in many areas, including X-ray crystallography \cite{millane1990phase,harrison1993phase}, coherent diffraction imaging \cite{chapman2006,miao1999extending,miao1998phase},  quantum mechanics \cite{Pauli,Philoso},
 astrophysics \cite{fienup1987phase}, optics \cite{shechtman2015phase,walther1963question}, and blind channel estimation \cite{baykal2004blind}, to name just a few.
 %In particular, the theoretical progress on the Fourier phase retrieval problem can lead to improvements to minimal informational completeness of the pure-state quantum tomography problem.

Although the Fourier phase retrieval problem is highly ill-posed in general, it has been shown theoretically that almost all the signals can be recovered uniquely if some additional information is available.  For instance, if one entry of the underlying signal is known in advances \cite{bendory2017fourier,beinert2015ambiguities} or if the underlying signal is sparse \cite{jaganathan2017sparse}, then the uniqueness of solutions is ensured for almost all signals.  Over the past few decades, many methods have been proposed for solving Fourier phase retrieval problem, such as the  Gerchberg-Saxton (GS) algorithm \cite{gerchberg1972practical}, the Fienup algorithm \cite{fienup1987reconstruction}, and the hybrid input-output (HIO) algorithm \cite{fienup1982phase}. However, the performance of those algorithms is not well-understood. Even in the absence of noise, those algorithms are not guaranteed to converge to the true solution. To the best of our knowledge, there is no method that provably recovers signals from their Fourier magnitude, even for multidimensional signals.
In this paper, we aim to give a theoretical understanding of algorithms for solving the Fourier phase retrieval problem.
Particularly, we   focus on the question:

{\em Is there any algorithm which possesses a linear convergence rate for Fourier phase retrieval, even when the exact value of the first entry of signal is known in advance?}

\subsection{Fourier phase retrieval}
Fourier phase retrieval problem, which seeks to reconstruct a signal from measurements of its Fourier magnitudes, has
been topics of intensive research over the past few decades.
Two fundamental questions of the Fourier phase retrieval involve  uniqueness of solutions and reconstruction algorithms.

 For the uniqueness, it is easy to check from \eqref{eq:probset} that the multiplication with a unimodular factor, the translation, and the conjugate reflection do not change the Fourier magnitudes. For this reason, the three fundamental forms of solutions are called trivial ambiguities, which are of minor interest. Beside the trivial ambiguities, it has been shown that there could be $2^{N-2}$ nontrivial solutions which strongly differ from the true signal \cite{bendory2017fourier,beinert2015ambiguities}.  In order to evaluate a meaningful solution of the Fourier phase retrieval, one needs to pose appropriate prior conditions to enforce the uniqueness of solutions.
 For instance, Beinert and Plonka \cite{beinert2018enforcing} proved that if  one absolute value of the true signal is known beforehand, then {\em almost} all signals can be uniquely recovered, and if the left $\lceil N/2 \rceil$ end-points are known in advance, then all signals can be uniquely determined. Another popular  way to enforce the uniqueness is to use the sparsity of the true signals. It has been shown that  \cite{jaganathan2017sparse} almost all $k$-sparse signals with $k<N$ can be uniquely recovered from their Fourier magnitudes up to rotations. However, it does not give any deterministic rule for which a sparse signal could be uniquely recovered from its Fourier magnitude. Later, Ranieri-Chebira-Lu-Vetterli proved that if a sparse signal is collision-free, then it can be uniquely recovered \cite{ranieri2013phase}. Finding the maximum collision-free set is also of independent interest in information theory \cite{bekir2007there} and discrete geometry \cite{senechal2008point}.  Others techniques to guarantee the uniqueness of solutions include the minimum phase method \cite{hayes1980signal,huang2016phase}, exploiting deterministic masks \cite{jaganathan2015phase,jaganathan2019reconstruction}, STFT measurements \cite{jaganathan2016stft,nawab1983signal,bendory2017non}, and the FROG method\cite{bendory2020signal}.
 We refer the reader to \cite{bendory2017fourier} for a recent survey and new results on the uniqueness of the Fourier phase retrieval problem.

For the reconstruction algorithms, the oldest method for recovering a signal from its Fourier magnitude  is given by Gerchberg and Saxton in 1972 \cite{gerchberg1972practical}. It is an alternating projection algorithm that iterates between the time and Fourier domains to match the measured Fourier magnitude and known constrains.  In 1982, Fienup extended the idea of the Gerchberg-Saxton method  to a variety of phase retrieval settings \cite{fienup1987reconstruction}, such as support, non-negativity and sparsity constrains. Later, a more popular algorithm, hybrid input-output method, is proposed for solving Fourier phase retrieval with additional information in the time domain \cite{gerchberg1972practical}.  More recently, algorithms based on convex relaxation techniques have been investigated for solving the Fourier phase retrieval problem, such as PhaseLift \cite{PRviamatrix,huang2016phase}.   However, there are no known optimality results for this approach.
We refer the interested reader to \cite{bauschke2003hybrid,chen2007application} for accounts of algorithms of Fourier phase retrieval. We shall emphasize that, to the best of our knowledge, all the algorithms listed above are heuristic in nature, and the theoretical understanding of those algorithms is limited. Even in the absence of noise, those algorithms are not guaranteed to converge to the true solution.

\subsection{Our contributions}
As stated before, there is a fundamental gap between the theory and algorithms for Fourier phase retrieval. Although all (or almost) signals can be determined uniquely by their Fourier magnitude with some additional information, there is no algorithm provably recovers the signals. All the existing methods are heuristic in nature and lacks the theoretical understanding.   In this paper, we give an intrinsic characterization of the algorithms for Fourier phase retrieval and show that it is impossible to propose an algorithm which possesses linear convergence rate for Fourier phase retrieval problem. For convenience, we restate the Fourier phase retrieval problem considered in this paper as follows (the details will be given in the next section):
  \begin{enumerate}
\item[] {\bf Fourier phase retrieval }:
Given $R(\omega)\geq 0, \omega\in [0,2\pi)$ and  a complex number $x_0\in \C$ ,
%and
%\begin{equation*}
%R(\omega)=e^{-i \omega(N+1)}\cdot  r(N-1) \prod_{n=1}^{N-1} (e^{i\omega }-\gamma_n) (e^{i\omega }- \wb{\gamma_n}^{-1})
%\end{equation*}
%for some $r(N-1) \neq 0$ and  $\gamma_n \neq 0, n=1,\ldots, N-1$.
find
\begin{equation}\label{eq:loss0}
\vx \in \argmin{\vz=(z(0),\ldots,z(N-1))\in \C^N}  \ell(|{{\z}}(\omega)|, R(\omega))\quad {\rm s.t.}\quad z(0)=x_0.
\end{equation}
Here, $\ell(|{{\z}}(\omega)|, R(\omega))$ is a nonnegative loss function which vanishes only when $\abs{{\z}(\omega)}^2= {R(\omega)}$.
    \end{enumerate}

We assume that ${\mathcal A}:={\mathcal A}(R(\omega), x_0)$ is an algorithm for solving Fourier phase retrieval (\ref{eq:loss0}), which can produce a sequence
$\dkh{\vx_0,\vx_1,\ldots}\subset \C^N$  with $\lim\limits_{m \rightarrow \infty}\vx_m=\vx$.
Here,  $\vx$ is a solution to  (\ref{eq:loss0}). We also assume that $\vx_m$ can be obtained by $\vx_0,\ldots, \vx_{m-1}$ in polynomial time. In this paper,  we say that ${\mathcal A}(R(\omega), x_0)$  is a {\em linear algorithm } if, for arbitrary $\epsilon>0$,  it can reach $\epsilon$-precision in $\mbox{Poly}(N) \log(1/\epsilon)$ steps, namely, the estimator $\vx_m \in \C^N$ satisfies
\[
{\norm{\vx_{m}-\vx}}\leq \epsilon,
\]
provided $m \ge \mbox{Poly}(N) \log(1/\epsilon)$. Here, $\mbox{Poly}(N)$ denotes a polynomial function with respect to $N$.

For example, if the output sequence  $\{\vx_1,\vx_2,\ldots\}$ of the algorithm  ${\mathcal A}$
satisfying
\[
 {\norm{\vx_j-\vx}}\leq \rho \cdot {\norm{\vx_{j-1}-\vx}}
\]
provided $j$ is large enough, where $\rho\in (0,1)$ is a constant,
then ${\mathcal A}$ is a linear algorithm.

We next state our main result:
\begin{theorem}\label{th:main}
If there exists  a linear algorithm for solving Fourier phase retrieval, then
 $\mathcal{ P}=\mathcal{NP}$.
\end{theorem}

Theorem \ref{th:main} shows that there is no linear algorithm to reconstruct a signal  from its discrete Fourier transform magnitude and the initial value
provided  $\mathcal{ P}\neq \mathcal{NP}$. Therefore,  the gap between the theory and algorithms of Fourier phase retrieval can not be reduced in general. The proof of this theorem is based on the relationship between Fourier phase retrieval and Product Partition problem. Specifically, we prove that if Fourier phase retrieval can be solved by a linear algorithm then Product Partition problem can be solved in polynomial time, which contradicts the fact that Product Partition problem is NP-complete in the strong sense.

For solving the Fourier phase retrieval problem, the loss function $\ell(|{\z}|, R)$ is employed.
The loss function  $\ell(|{ \z}|, R)=\sum_{j=1}^M (|{\z}(\omega_j)|-\sqrt{R(\omega_j)})^2$ is  employed in the
Gerchberg-Saxton (GS) algorithm \cite{gerchberg1972practical}, the Fienup algorithm \cite{fienup1987reconstruction} and the hybrid input-output (HIO) algorithm \cite{fienup1982phase}, while the loss function
$\ell(|{\hat \vz}|, R)=\sum_{j=1}^M (|\hat{\vz}(\omega_j)|^2-{R(\omega_j)})^2$ is employed in the Wirtinger flow method\cite{WF}.
Here,  $\{\omega_j\}_{j=1}^M\subset [0,2\pi)$ are discrete samples.

\begin{remark}
In Theorem \ref{th:main}, we show that the Fourier phase retrieval problem cannot be solved in $ \mbox{Poly}(N) \log(1/\epsilon)$ time to reach the $\epsilon$-precision by solving \eqref{eq:loss0}. Actually, using the same arguments, we could prove that there is no hope to propose an optimization algorithm which could reach $\epsilon$-precision within  $\mbox{Poly}(N\log(1/\epsilon))$ time. Note that the SDP program can be solved by interior point methods within $ \mbox{Poly}(N) \log(1/\epsilon)$  time to reach $\epsilon$-precision \cite{SDP}. Therefore, our results demystify the phenomenon that why no optimality results have been established for the natural convex SDP relaxation of the Fourier phase retrieval problem, such as PhaseLift given in \cite{PRviamatrix,huang2016phase}.
\end{remark}

\subsection{Organization}
The paper is organized as follows.  In Section 2,  we give a complete characterization of the non-trivial ambiguities for Fourier phase retrieval and then introduce several basic terminologies in computational complexity theory. In particular, we show that  the main challenge  in Fourier phase retrieval is to decide which zero should be selected from each zero pair. In Section 3,  we give the proof of our main result, i.e.,
Theorem  \ref{th:main}.  Finally, a brief discussion is presented in Section 4.

\section{Preliminaries}

\subsection{Characterization of Fourier phase retrieval}

For Fourier phase retrieval, as given in \eqref{eq:probset},  the measurements that we obtain are
\begin{equation*}
|\hx(\omega)|:=\abs{\sum_{n=0}^{N-1} x(n) e^{-i\omega n}}, \quad \omega \in [0,2\pi).
\end{equation*}
For the convenience, we set $x(k)=0$ for $k<0$ and $k\ge N$. Then it  gives that
\[
R(\omega):=|\hx(\omega)|^2=\sum_{n\in  \Z} \sum_{k \in \Z} x(n) \wb{x(k)} e^{-i\omega (n-k)}= \sum_{n\in  \Z}  r(n) e^{-i\omega n},
\]
where $\vr=(r(n))_{n\in \Z}$ is the autocorrelation signal of $\vx$ which is defined by
\begin{equation} \label{eq:auror}
r(n):=\sum_{k=0}^{N-1} \overline{x(k)}  x(k+n), \quad  n\in \Z.
\end{equation}
Note that $r(n)=0$ for all $|n| \ge N$ and $\vr$ can be obtained directly by the inverse Fourier transform of the signal's Fourier intensity.
Therefore, Fourier phase retrieval is equivalent to the reconstruction of $\vx\in \C^N$ from its autocorrelation $\vr \in \C^{2N-1}$.

It is well-known that Fourier phase retrieval is not uniquely solvable. There exists the so-called trivial ambiguities caused by the rotation, the translation and the conjugate reflection. From a physical point of view, the trivial ambiguities are acceptable. However, beside the trivial ambiguities, there are also a series of nontrivial solutions. In order to classify those nontrivial solutions, we  reformulated the Fourier phase retrieval problem as follows: For a given real nonnegative autocorrelation polynomial
\[
R(\omega):=\sum_{n=-N+1}^{N-1} r(n) e^{-i\omega n},
\]
find all trigonometric polynomials
\[
\hx(\omega):=\sum_{n=0}^{N-1} x(n) e^{-i\omega n}
\]
such that
\[
|\hx(\omega)|^2=R(\omega).
\]

The following lemma gives a characterization of nontrivial ambiguities of Fourier phase retrieval, which can be found in \cite[Theorem 2.4]{beinert2015ambiguities}. To keep the paper self contained, we give a brief proof here.

\begin{lemma}\label{le:nontricha}
Let $R(z):=\sum_{n=-N+1}^{N-1} r(n) z^{-n}$ be the  $Z$-transform of a vector $\vr:=(r(-N+1),\ldots,r(N-1) ) \in\C^{2N-1}$ with $r(-n)=\wb{r(n)}$ for $n=0,\ldots, N-1$, and $r(N-1)\neq 0$. Suppose that $R(z) \ge 0$ for all $z\in \C$ on the unit circle. Then $R(z)$ can be written in the following form:
\begin{equation}\label{eq:RZ}
R(z)=z^{-N+1}\cdot  r(N-1)\cdot \prod_{n=1}^{N-1} (z-\gamma_n) (z- \wb{\gamma_n}^{-1})
\end{equation}
for some $\gamma_n \neq 0, n=1,\ldots, N-1$. Furthermore, for any $\vx:=(x(0),\ldots, x(N-1))\in \C^{N}$, if the $Z$-transform $X(z):=\sum_{k=0}^{N-1} x(k) z^{-k}$ of $\vx$ satisfying $|X(z)|^2=R(z)$ for all $z\in \C$ on the unit circle, then we have
\begin{equation} \label{eq:ztranx}
X(z)=z^{-N+1} \cdot e^{i\alpha} \cdot |r(N-1)|^{\frac12}\cdot \prod_{n=1}^{N-1} |\beta_n|^{-\frac12} \xkh{z-\beta_n},
\end{equation}
where $\alpha \in [-\pi, \pi)$ and the roots $\beta_n\in \dkh{\gamma_n, \wb{\gamma_n}^{-1}}$.
\end{lemma}
\begin{proof}
Set $S(z):=z^{N-1} R(z)$. Then $S(z)$ is a complex algebraic polynomial of degree $2N-2$, namely,
\[
S(z)=r(0) z^{N-1} +\sum_{n=1}^{N-1} \wb{r(n)} z^{N+n-1} +\sum_{n=1}^{N-1} r(n) z^{N-n-1}.
\]
Assume that $S(\gamma_n)=0$, namely, $\gamma_n$ is a root of  $S(z)$.   The root $\gamma_n\neq 0$ due to the constant term $\wb{r(N-1)}\neq 0$.
We claim that $S(\wb{\gamma_n}^{-1})=0$, i.e.,  $\wb{\gamma_n}^{-1}$ is also a root of $S(z)$. Indeed,
a simple calculation shows that
\begin{eqnarray*}
S( \wb{\gamma_n}^{-1}) &= & r(0) \wb{\gamma_n}^{-N+1} +\sum_{n=1}^{N-1} \wb{r(n)} \wb{\gamma_n}^{-N+1-n} +\sum_{n=1}^{N-1} r(n) \wb{\gamma_n}^{-N+1+n}\\
&=& \wb{\gamma_n}^{-2N+2 }\xkh{ \wb{r(0) } \wb{\gamma_n}^{N-1} +\sum_{n=1}^{N-1} \wb{r(n)} \wb{\gamma_n}^{N-1-n} +\sum_{n=1}^{N-1} r(n) \wb{\gamma_n}^{N-1+n} }\\
&=& \wb{\gamma_n}^{-2N+2 } \wb{S(\gamma_n)}=0,
\end{eqnarray*}
where we use $S(\gamma_n)=0$ in  the last equality. This gives the claim.  Therefore, all roots of $S(z)$ not lying on the unit circle occur in pairs $\xkh{\gamma_n, \wb{\gamma_n}^{-1}}$. For the case where the root $\gamma_n$ lying on the unit circle, the multiplicity must be even due to fact that  $R(z) \ge 0$ for all $z\in \C$ on the unit circle. Hence,  $R(z)$ can be written as
\[
R(z)=z^{-N+1}\cdot  r(N-1) \prod_{n=1}^{N-1} (z-\gamma_n) (z- \wb{\gamma_n}^{-1})
\]
for some $\gamma_n \neq 0, n=1,\ldots, N-1$.

Furthermore,  for any $z\in \C$ on the unit circle, we have
\begin{eqnarray*}
|(z-\gamma_n) (z- \wb{\gamma_n}^{-1})| &=& |\wb{z}| \cdot |(z-\gamma_n) (z- \wb{\gamma_n}^{-1})| \\
&=& |z-\gamma_n|\cdot  |\wb{\gamma_n}^{-1} \wb{z}-1| \\
&=&  |{\gamma_n}|^{-1}\cdot |z-\gamma_n|^2.
\end{eqnarray*}
Similarly, we can obtain $|(z-\gamma_n) (z- \wb{\gamma_n}^{-1})| = |\wb{\gamma_n}|  \cdot |z-\wb{\gamma_n}^{-1}|^2$ for $z$ on the unit circle.  Recall that $|X(z)|^2=R(z)$ and $R(z)\ge 0$ for  all $z\in \C$ on the unit circle. It then gives $|X(z)|=\sqrt{|R(z)|}$ for $z$ on the unit circle.  Combining with the fact that  $z^{-N+1} X(z)$ is a complex algebraic polynomial of degree $N-1$, we immediately have
\begin{equation} \label{eq:xz}
X(z)=z^{-N+1} \cdot e^{i\alpha} |r(N-1)|^{\frac12} \prod_{n=1}^{N-1} |\beta_n|^{-\frac12} \xkh{z-\beta_n}
\end{equation}
for $z$ on the unit circle, where $\beta_n\in \dkh{\gamma_n, \wb{\gamma_n}^{-1}}$.

Finally, noting that $z^{-N+1} X(z)$ is analytic, we obtain that  \eqref{eq:xz} holds for all $z \in \C$, which completes the proof.

\end{proof}

Lemma \ref{le:nontricha} demonstrates that the main problem in the recovery of $\vx$ from its Fourier imagnitudes is  how to select the zeros $\beta_n$ from the zero pairs $\dkh{\gamma_n, \wb{\gamma_n}^{-1}}$ to form $X(z)$.
%since the multiplication by $e^{i\alpha}$ corresponds to the trivial rotation ambiguity.
Note that there are $2^{N-1}$ possible choices, which leads to $2^{N-2}$ possible nontrivial solutions of the considered phase retrieval problem, after removing the trivial conjugate reflection ambiguity.

% consider the  $Z$-transform of $\vx \in\C^N$, namely,
%\begin{equation} \label{eq:ztranx}
%X(z):=\sum_{k=0}^{N-1} x(k) z^{-k}=x(0)z^{-(N-1)}(z-\beta_1)\cdots (z-\beta_{N-1}).
%\end{equation}
%Hence, the reconstruction of $\vx$ can be recast as determining the zeros $\beta_1,\ldots,\beta_{N-1}$ and $x(0)$.
%Let $R(z)$ be the $Z$-transform of the autocorrelation $\vr \in \C^{2N-1}$ given in \eqref{eq:auror}.  If we restrict $R(z)$ and $X(z)$ on the unit circle then their relationship is
%\begin{eqnarray*}
% R(z)&=&\abs{X(z)}^2\\
%&=& X(z) \wb{X(z)} \\
%&=& X(z) \wb{X(1/\bar{z})} \\
%&=& \abs{x(0)}^2 z^{-(N-1)}  (z-\beta_1)(1-\wb{\beta_1}z)\cdots (z-\beta_{N-1})(1-\wb{\beta_{N-1}}z).
%\end{eqnarray*}
% Observe that $z^{N-1}R(z)$ is a polynomial with zeros $\beta_1,\wb{\beta_1}^{-1},\ldots, \beta_{N-1}, \wb{\beta_{N-1}}^{-1}$, and its zeros appear in pairs of the form $\dkh{\beta_j,\wb{\beta_j}^{-1}}$. Therefore,
%  It is easy to see that

In order to evaluate a meaningful solution of the Fourier phase retrieval, one needs to pose appropriate prior conditions to enforce uniqueness of solutions.  One way to achieve this goal is to use additionally known values of some entries, such as the value of the first entry $x(0)$.
If we know the value $x(0)$ it then follows from \eqref{eq:auror} that  the value $x(N-1)$ can be obtained by
\[
x(N-1)=\frac{r(N-1)}{\wb{x(0)}}.
\]
Noting that $z^{N-1}X(z)=\sum_{k=0}^{N-1} x(k)z^{N-1-k}$ and $X(\beta_n)=0, n=1,\ldots,N-1$,
 we obtain that
\begin{equation} \label{eq:prodcon}
\prod_{j=1}^{N-1} (-\beta_j) =\frac{x(N-1)}{x(0)}=\frac{r(N-1)}{|x(0)|^2}.
\end{equation}
Actually, it has been proved that the constrain \eqref{eq:prodcon} can ensure the uniqueness of solutions for almost all signal $\vx \in \C^N$ \cite{bendory2017fourier,beinert2015ambiguities}.

\subsection{Complexity Theory and Product Partition Problem}
In order to prove our main result, we need the famous result in combinatorics that  Product Partition Problem is NP-complete in the strong sense \cite{garey1979computers,ng2010product,kovalyov2010generic}.   Product Partition problem can be formulated as follows.
\begin{enumerate}
%\item[] {\bf Partition}: Given $n$ positive integer number $u_1,\ldots,u_n$, is there a  subset $\Gamma \subset \dkh{1,\ldots,n}$ such that $\sum_{j\in\Gamma} u_j= \sum_{j\in\Gamma^c} u_j$?
\item[] {\bf Product Partition}: Given $n$ positive integer number $u_1,\ldots,u_n$,  is there a subset $\Gamma \subset \dkh{1,\ldots,n}$ such that $\prod_{j\in \Gamma} u_j=\prod_{j \in \Gamma^c} u_j $?
\end{enumerate}

Throughout the paper, we use the basic terminology  in computational complexity theory \cite{garey1978strong}. For instance, the complexity class $\mathcal{ P}$ is the set of problems which can be solved in an amount of time that is bounded by a polynomial of the binary coding of the input.
A pseudo-polynomial time algorithm is an algorithm that runs in time bounded by a polynomial of the number of input and the magnitude of the largest number of the input.
 If a problem is NP-complete, then it does not have a polynomial time algorithm, unless $\mathcal{ P}=\mathcal{NP}$. Similarly, if a problem is NP-complete in the strong sense then it does not have a pseudo-polynomial time algorithm, unless $\mathcal{ P}=\mathcal{NP}$.  It is widely believed  $\mathcal{ P} \neq \mathcal{NP}$ \cite{fortnow2009status}, but the proof is still open.

%The following lemma plays a key role in the proof of main result.
%
%\begin{lemma} \label{le:solupar}
%The problem of finding a solution of Partition (Product Partition) remains NP-complete (NP-complete in the strong sense, respectively), even if the solution of Partition (Product Partition) exists.
%\end{lemma}
%\begin{proof}
%Denote the problem of finding a solution of Partition as ``SolutionPartition''.  We could prove that SolutionPartition  is NP-complete by reduction from Partition. To this end, assume that we could solve SolutionPartition in $T(L)$ steps, where $L$ is the binary coding of the input and $T(L)$ is a polynomial function of $L$. Then for any Partition problem,  we could regard its input as the input of SolutionPartition and stop after $T(L)$ steps. If the algorithm returns an equal sum partition, then the answer to the original Partition problem is ``yes'', otherwise the answer is ``no''.  Since Partition problem is NP-complete,
%it gives that SolutionPartition is also NP-complete.
%
%Using the same argument, we could prove that the problem of finding a solution of Product Partition is NP-complete in the strong sense. This completes the proof.
%\end{proof}
%

\section{Proof of Theorem \ref{th:main}}

The aim of this section is to prove Theorem \ref{th:main} showing that there is no linear algorithm for solving the Fourier phase retrieval problem. The idea of the proof is: if there is an algorithm with linear convergence rate for solving the Fourier phase retrieval problem, then we could propose another algorithm such that it can give an answer to Product Partition problem in  polynomial time.  This contradicts the famous result that Product Partition is NP-complete, which leads to our result.
To begin with, we need the following auxiliary lemma.

\begin{lemma} \label{le:key2}
Assume that $u_1,\ldots,u_{N} \ge 2$ are positive integers. Assume that  there is an index set $\Gamma \subset \dkh{1,\ldots, N-1}$ such that  $\prod_{k\in \Gamma} u_k =u_{N}  \cdot \prod_{k\in \Gamma^c} u_k $.
 Set $u_{\max}:=\max\{u_1,\ldots,u_{N-1}\}$.
 Let  $\vx:=(x(0),\ldots, x(N-1)) \in \C^N$ be a vector constructed according to the following two conditions:
\begin{itemize}
\item[(1)] The initial value $x(0)=u_{\max}^{N-1}$.
\item[(2)] The roots of the $Z$-transform of $\vx$ are
\begin{equation*}
\beta_k =\left\{\begin{array}{ll} -{u_k}, &\quad  \mbox{if} \quad k\in \Gamma\\
-\frac1{u_k}, &\quad  \mbox{if} \quad k\in \Gamma^c\\
\end{array} \right.,
\end{equation*}
where the $Z$-transform of $\vx$ is defined  as
\begin{equation*}
X(z):=\sum_{k=0}^{N-1} x(k) z^{-k}=x(0)z^{-(N-1)}(z-\beta_1)\cdots (z-\beta_{N-1}).
\end{equation*}
\end{itemize}
Assume that $\vx_m \in \C^N$ is a vector  which satisfies  $ \norm{\vx_m-\vx} \le  u_{\max}^{-2N}$.
We use  $X_m(z)$ to denote  the $Z$-transform of $\vx_m$.
Then  the followings  hold:
\begin{itemize}
\item[(i)] If $1/\beta_k$ is not a zero of $X(z)$ then
\[
\abs{ X_m(1/\beta_k)} \ge \abs{ X_m(\beta_k)}+c_0
 \]
 holds for a constant $c_0:=1-2u_{\max}^{-N}$.
 \item[(ii)] If $1/\beta_k$ is a zero of $X(z)$ then
 \[
 \max\dkh{\abs{X_m(\beta_k)},\abs{ X_m(1/\beta_k)} } \le u_{\max}^{-N}.
 \]
\end{itemize}

\end{lemma}

\begin{proof}
  To prove part (i),  we need to present an upper bound for $\abs{ X_m(\beta_k)}$ and a lower bound for $\abs{ X_m(1/\beta_k)} $.  Consider the following two cases:

{\bf Case 1:} $k\in \Gamma$,  namely, $\beta_k=-{u_k}$ is the zero of $X(z)$. Recall that
\[
X(z)=\sum_{n=0}^{N-1} x(n) z^{-n} \quad \mbox{and} \quad X_m(z)=\sum_{n=0}^{N-1} x_m(n) z^{-n}.
\]
Since $X(-u_k)=0$,  the upper bound for $\abs{ X_m(\beta_k)}$ can be given as
\begin{equation}\label{eq:lin1}
\begin{aligned}
\abs{ X_m(\beta_k)}:=\abs{X_m(-{u_k}) } & =  \abs{X_m(-{u_k}) -X(-{u_k}) }  \\
&= \abs{ \sum_{n=0}^{N-1} (-1)^n u_k^{-n } \xkh{ x_m(n)-x(n) } }  \\
&\le  \norm{\vx_m-\vx} \cdot \sqrt{\sum_{n=0}^{N-1} u_k^{-2n } }   \\
&\le  \norm{\vx_m-\vx} \cdot \sqrt{ \frac1{1-1/u_k^2} }   \\
&\le  2 \norm{\vx_m-\vx} \\
&\le   u_{\max}^{-N},
\end{aligned}
\end{equation}
where the first inequality follows from the Cauchy-Schwarz inequality,  the third inequality follows from the fact that  $u_k\ge 2$ are positive integers, and the last inequality follows from the assumption of
$\norm{\vx_m-\vx} \le u_{\max}^{-2N}$.

The lower bound for $\abs{ X_m(1/\beta_k)} $ can be formulated as
\begin{eqnarray} \label{eq:lintri}
\abs{ X_m(1/\beta_k)} :=\Big |X_m(-\frac1 {u_k})\Big|  \ge \Big| X(-\frac1 {u_k})\Big| - \Big| X_m(-\frac1 {u_k})-X(-\frac1 {u_k})\Big|.
\end{eqnarray}
For the first term $|X(-\frac1 {u_k}) |$, recall  that $X(z)=x(0)z^{-(N-1)}(z-\beta_1)\cdots (z-\beta_{N-1})$ and the roots of $X(z)$ are
\begin{equation*}
\beta_k =\left\{\begin{array}{ll} -{u_k}, &\quad  \mbox{if} \quad k\in \Gamma\\
-\frac1{u_k}, &\quad  \mbox{if} \quad k\in \Gamma^c\\
\end{array} \right..
\end{equation*}
Therefore,   $u_j\neq u_k$  for any $j\in \Gamma^c$ since $k\in \Gamma$ and $1/u_k$ is not a zero of $X(z)$. Observe that the initial value $x(0)=u_{\max}^{N-1}$. Thus,  we have
\begin{equation} \label{eq:lin2}
\begin{aligned}
\Big| X(-\frac1 {u_k}) \Big| & =  u_{\max}^{N-1} u_k^{N-1} \prod_{j\in\Gamma} \abs{{u_j} - u_k^{-1}} \prod_{j\in\Gamma^c} \abs{u_j^{-1} -  {u_k}^{-1}}  \\
& =  u_{\max}^{N-1} \prod_{j\in\Gamma} \abs{u_j u_k - 1} \prod_{j\in\Gamma^c} \abs{\frac{u_k}{u_j} -  1}  \\
&\ge u_{\max}^{N-1}  \prod_{j\in\Gamma} u_j \prod_{j\in\Gamma^c} u_j^{-1}  \\
&= u_{\max}^{N-1} u_N   \\
&\ge 2^N,
\end{aligned}
\end{equation}
where the first inequality follows from the fact that $u_j \ge 2$ are positive integers and $u_j\neq u_k$   for any $j\in \Gamma^c$, and the last equality follows from  $\prod_{t\in \Gamma} u_t =u_{N} \prod_{t\in \Gamma^c} u_t $.

For the second term of \eqref{eq:lintri}, we have
\begin{equation} \label{eq:lin3}
\begin{aligned}
\abs{X_m(-\frac{1}{u_k})-X(-\frac{1}{u_k})}  = & \abs{ \sum_{n=0}^{N-1} (-u_k)^n \xkh{ x_m(n)-x(n) } }  \\
\le & \norm{\vx_m-\vx} \cdot \sqrt{\sum_{n=0}^{N-1} u_k^{2n} }  \\
\le &  \norm{\vx_m-\vx}  \cdot \sqrt{\frac{u_k^{2N}-1}{u_k^2-1}}    \\
\le &   \norm{\vx_m-\vx} \cdot u_k^N   \\
\le &   u_{\max}^{-N},
\end{aligned}
\end{equation}
where the last inequality follows from the condition  $\norm{\vx_m-\vx} \le u_{\max}^{-2N}$.
Putting \eqref{eq:lin1}, \eqref{eq:lintri}, \eqref{eq:lin2} and \eqref{eq:lin3} together, we immediately  obtain
\[
\abs{X_m(1/\beta_k) } \ge  2^N-u_{\max}^{-N} \geq \abs{X_m(\beta_k) } +2^N -2u_{\max}^{-N}
\]
for $k\in \Gamma$. Here, the first inequality is obtained by \eqref{eq:lintri}, \eqref{eq:lin2} and \eqref{eq:lin3} and
the second inequality follows from \eqref{eq:lin1}.

{\bf Case 2:} $k\in \Gamma^c$,  namely, $\beta_k=-u_k^{-1}$ is the zero of $X(z)$.
For this case, we have
\begin{equation}\label{eq:case21}
\abs{X_m(\beta_k)}=\abs{X_m(-u_k^{-1}) }  =  \abs{X_m(-u_k^{-1}) -X(-u_k^{-1}) }  \le  u_{\max}^N\norm{\vx_m-\vx} \le  u_{\max}^{-N}
\end{equation}
due to the fact $\norm{\vx_m-\vx} \le  u_{\max}^{-2N}$. For the term $\abs{X_m(1/\beta_k) }$, we have
\begin{equation} \label{eq:case22}
\abs{X_m(1/\beta_k) }= \abs{X_m(-{u_k}) } \ge \abs{X(-{u_k})}- \abs{X_m(-{u_k})-X(-{u_k})}.
\end{equation}
A simple calculation shows that
\begin{equation}\label{eq:case23}
\begin{aligned}
\abs{X(-{u_k}) } &=  u_{\max}^{N-1} u_k^{-(N-1)} \prod_{j\in\Gamma} \abs{u_k-u_j} \prod_{j\in\Gamma^c} \abs{u_k- u_j^{-1}} \\
&\ge  \xkh{\frac{u_{\max}}{u_k}}^{N-1}\\
&\ge 1,
\end{aligned}
\end{equation}
where the first inequality comes from the fact $u_j$ are positive integers and $u_j\neq u_k$  for $j\in \Gamma$ due to  $k\in \Gamma^c$ and $1/u_k$ is not a zero of $X(z)$.   Another term in (\ref{eq:case22}) can be bounded by
\begin{equation}\label{eq:case24}
 \abs{X_m(-{u_k}) -X(-{u_k}) }  \le    2 \norm{\vx_m-\vx} \le  u_{\max}^{-N}.
\end{equation}
Combining the results above, we obtain that
\[
\abs{ X_m(1/\beta_k)} \overset{(\ref{eq:case22})}\ge
\abs{X(-{u_k})}- \abs{X_m(-{u_k})-X(-{u_k})} \overset{(\ref{eq:case23}), (\ref{eq:case24})}\geq 1-u_{\max}^{-N} \overset{(\ref{eq:case21})}\ge
\abs{ X_m(\beta_k)}+1-2u_{\max}^{-N}
\]
for $k\in \Gamma^c$.

In summary, for all $k=1,\ldots, N-1$, if $1/\beta_k$ is not a zero of $X(z)$ then
\[
\abs{ X_m(1/\beta_k)} \ge \abs{ X_m(\beta_k)}+c_0 .
 \]
Here,  $c_0:=1-2u_{\max}^{-N}$.

For part (ii) where both $\beta_k$ and $1/\beta_k$ are zeros of $X(z)$, we can obtain the conclusion by  the inequalities \eqref{eq:lin1} and \eqref{eq:lin3}.

\end{proof}

Now, we are ready to prove the main result of this paper .

%\begin{theorem} \label{th:mainlinear}
%There is no linear algorithm to reconstruct a signal $\vx \in \C^N$ from its discrete Fourier transform magnitude and the initial value $x(0)$, unless $\mathcal{ P}=\mathcal{NP}$.
%\end{theorem}
\begin{proof}[Proof of  Theorem \ref{th:main}]
Our main idea is to  show that if there is a linear algorithm for solving the Fourier phase retrieval problem, then Product Partition problem can be solved in polynomial time.  This contradicts the result that  Product Partition is NP-complete in the strong sense.  To do this, recall that the Fourier phase retrieval problem that we consider is
\begin{equation}\label{eq:loss22}
\vx \in \argmin{\vz=(z(0),\ldots,z(N-1))\in \C^N}  \ell(|{{\z}}(\omega)|, R(\omega))\quad {\rm s.t.}\quad z(0)=x_0,
\end{equation}
where  $x_0\in \C$ is the given initial value and
\begin{equation} \label{eq:Romega}
R(\omega)=e^{-i \omega(N+1)}\cdot  r(N-1) \prod_{n=1}^{N-1} (e^{i\omega }-\gamma_n) (e^{i\omega }- \wb{\gamma_n}^{-1})
\end{equation}
for some $r(N-1) \neq 0$ and  $\gamma_n \neq 0, n=1,\ldots, N-1$.   Here, $\ell(|{{\z}}(\omega)|, R(\omega))$ is a nonnegative loss function which vanishes only when $\abs{{\z}(\omega)}^2= R(\omega) $. Assume  that $\mathcal A:=\mathcal A(R(\omega), x(0))$ is the linear algorithm  which can output an estimator  $\vx_m\in \C^N$ satisfying
\begin{equation} \label{eq;quad}
\norm{\vx_m-\vx} \le \epsilon
\end{equation}
 in $\mbox{Poly}(N) \log(1/\epsilon)$  time.
%Recall that $\vx:=(x(0),\ldots, x(N-1))^\top \in \C^N$ is the target signal and we use $X(z)$ and $R(z)$ to denote the Z-transform of $\vx$, and the autocorrelation vector $\vr$, respectively.
 According to (\ref{eq:Romega}),
% $\abs{\hat{\vx}(\omega)}^2=R(\omega)$   can be obtained by the zero pairs   $\dkh{(\gamma_1,\wb{\gamma_1}^{-1}),\ldots, (\gamma_{N-1}, \wb{\gamma_{N-1}}^{-1})}$ of $R(z)$ and the value $r(N-1)$ in polynomial time.  Hence, we next
we could assume that the inputs of $\mathcal A$ are the values of $r(N-1)$ and  $x(0)$, and the zero pairs   $\dkh{(\gamma_1,\wb{\gamma_1}^{-1}),\ldots, (\gamma_{N-1}, \wb{\gamma_{N-1}}^{-1})}$.

For any Product Partition problem with positive integers $u_1,\ldots,u_N \ge 2$,  without loss of generality, we assume $u_{\max}:=\max\{u_1,\ldots,u_N\} \ge 3$. We claim that  Algorithm \ref{alg2} could give an answer to this product partition problem in polynomial time.  This is a contradiction, which leads to our conclusion.

It remains to prove the claim.
We first consider the case where  the Product Partition problem has a solution $\Gamma \subset \dkh{1,\ldots,N-1}$ such that
\begin{equation} \label{eq:orpart22}
\prod_{k\in \Gamma} u_k =u_{N} \prod_{k\in \Gamma^c} u_k.
\end{equation}
We next show that Algorithm \ref{alg2} can output the set $\Gamma$ in polynomial time.
Construct a target  signal $\vx:=(x(0),\ldots, x(N-1)) \in \C^N$ with  $x(0)=x_0:=u_{\max}^{N-1}$ and the  roots of the $Z$-transform of $\vx$ being
\begin{equation} \label{eq:zerosbeta22}
\beta_k =\left\{\begin{array}{ll} -{u_k}, &\quad  \mbox{if} \quad k\in \Gamma\\
{-u_k^{-1}}, &\quad  \mbox{if} \quad k\in \Gamma^c\\
\end{array} \right.,
\end{equation}
where the $Z$-transform  of $\vx$ is defined as
\begin{equation} \label{eq:ztranx222}
X(z):=\sum_{k=0}^{N-1} x(k) z^{-k}=x(0)z^{-(N-1)}(z-\beta_1)\cdots (z-\beta_{N-1}).
\end{equation}
\IncMargin{1em}
\begin{algorithm}[H]
 \SetAlgoLined
 \SetKwData{Return}{return}
\SetKwInOut{Input}{Input}\SetKwInOut{Output}{Output~}
	\caption{Solving Product Partition problem based on Fourier phase retrieval algorithm $\mathcal A$}
	\label{alg2}
\Input{Positive integers $u_1, u_2,\ldots, u_N$}
\BlankLine
    Initialization: $U:=\left\{ u_1,\ldots,u_{N-1} \right\}$; Boolean variable $b=\textbf{TRUE}$; Initial value $x(0)=u_{\max}^{N-1}$ with $u_{\max}:=\max\xkh{u_1,\ldots,u_{N-1}}$ \;
		\While{$b=\textbf{TRUE}$ and $U\neq \emptyset$ } {
		Define two sets $\Gamma_1=\emptyset, \Gamma_2=\emptyset$ \;
		 Rewrite $U:=\dkh{u_{k_1},\ldots, u_{k_{N-1}}}$ and then compute $\vx_m \in \C^N$ by the algorithm $\mathcal A$ based on initial value  $x(0)$, $r(N-1):=|x(0)|^2 u_{N}$, and  the zero pairs $\dkh{(-{u_{k_1}},{-u_{k_1}^{-1}}),\ldots, (-{u_{k_{N-1}}},-u_{k_{N-1}}^{-1})}$.  Here,  $\vx_m$ is the output of the algorithm $\mathcal A$ after $m=\mbox{Poly}(N)\log u_{\max}$ steps\;
		\For{ $j=1,\ldots, N-1$ } {
		 Compute $X_m(-{u_{k_j}})$ and $X_m({-u_{k_j}^{-1}})$. Here, $X_m(z)$ is the $Z$-transform of the iteration $\vx_m$ \;
		  \eIf {$\max\xkh{ X_m(-{u_{k_j}}), X_m({-u_{k_j}^{-1}})} \le 1/4$  } {
\tcc{both $-{u_{k_j}}$ and ${-u_{k_j}^{-1}}$  are roots of $X(z)$.}
\eIf{there exits $k_\ell$ such that $u_{k_\ell}=u_{k_j}$ and $k_\ell\neq k_j$} {
		 $U\leftarrow U \setminus \dkh{u_{k_\ell}, u_{k_j}}$, $N \leftarrow |U|+1$, initial value $x(0)\leftarrow u_{\max}^{N-1}$\;
		 {\bf break} and  {\bf go} Step 2 \;		
}{
 \tcc {For the case where both $-{u_{k_j}}$ and ${-u_{k_j}^{-1}}$  are roots of $X(z)$, if the Product Partition has a solution, then there muse be $k_\ell$ such that $u_{k_\ell}=u_{k_j}$ and $k_\ell\neq k_j$.}
{\bf return} { ``The product partition problem has a solution''}
}
}{
     \eIf  {$\abs{ X_m(-u_{k_j}^{-1})} \ge \abs{ X_m(-u_{k_j})}+ 3/4$} {
     $\Gamma_1\leftarrow \Gamma_1\cup \dkh{k_j}$ \;
}{
     $\Gamma_2\leftarrow \Gamma_2\cup \dkh{k_j}$ \;
}
}
}
 $b=\textbf{FALSE}$ \;
		}
	  Compute $\mbox{quot}= \frac{\prod_{j\in \Gamma_1} u_j }{ \prod_{j\in \Gamma_2} u_j}$ \;
	  \BlankLine
	  \Output {If $U= \emptyset$ or quot$\neq u_N$ {\bf return} ``The product partition problem has no solution'' ;
otherwise {\bf return}  ``The product partition problem has a solution $\Gamma:=\Gamma_1$'' }
\end{algorithm} \DecMargin{1em}
\newpage

For this signal $\vx$,  we apply the  algorithm $\mathcal A$  to recover it from the  initial value  $x(0):=u_{\max}^{N-1}$, the zero pairs $\dkh{(\gamma_1,\wb{\gamma_1}^{-1}),\ldots, (\gamma_{N-1}, \wb{\gamma_{N-1}}^{-1})}:=\dkh{(-{u_{1}},{-u_{1}^{-1}}),\ldots, (-{u_{{N-1}}},-u_{{N-1}}^{-1})}$,  and $r(N-1):=|x(0)|^2 u_{N}$, i.e., taking
  \begin{equation*}
R(\omega):=e^{-i \omega(N+1)}\cdot  |x(0)|^2\cdot u_{N}\cdot  \prod_{n=1}^{N-1} (e^{i\omega }+u_n) (e^{i\omega }+{u_n}^{-1}).
\end{equation*}
  It then follows from \eqref{eq;quad} that the iteration $\vx_m$ given by $\mathcal A$ after $m= \mbox{Poly}(N) \log u_{\max}$ steps
 obeys
\begin{equation}\label{eq:sufficlose22}
\norm{\vx_m-\vx} \le  u_{\max}^{-2N}.
\end{equation}

%provided $m=  O(N\log u_{\max})$. Here,  we assume the initial point $\vx_0$ satisfying $\norm{\vx_0} \le c_0 \norm{\vx}$ for a universal constant in the third inequality, and we use the fact that
%\[
%\norm{\vx} \le \normone{\vx} \le x(0) (1+ 2^N u_{\max}^N) \le u_{\max}^{3N}
%\]
%in the last inequality, where we apply the Vieta's formulas in \eqref{eq:ztranx222} to deduce
%\[
%|x(k)| \le \left( \begin{array}{l} n\\k \end{array} \right) u_{\max}^k x(0) \quad \mbox{for all} \quad k=1,\ldots,N-1.
%\]
%in the above inequality.
Next, we divide the proof into the following two cases:

{\bf Case 1:} $1/\beta_k$ are not  zeros of $X(z)$ for all $k=1,\ldots, N-1$. From Lemma \ref{le:key2} and \eqref{eq:sufficlose22}, we obtain that
\begin{equation}\label{eq:Xmcase1}
\abs{ X_m(1/\beta_k)} \ge \abs{ X_m(\beta_k)}+1-2u_{\max}^{-N} \ge \abs{ X_m(\beta_k)}+ \frac34
 \end{equation}
holds for all $k=1,\ldots, N-1$, provided $m= \mbox{Poly}(N) \log u_{\max}$. Here, $X_m(z)$ is the $Z$-transform of iteration $\vx_m$.
This gives a rule to determine which zero should be selected from $\dkh{-{u_{k}},{-u_{k}^{-1}}}$ to form  $\beta_{k}$ for all $k=1,\ldots, N-1$
(see Line 15 in Algorithm \ref{alg2}). Moreover, if  $\beta_{k}=-{u_{k}}$ then $k\in \Gamma$, otherwise $k\in \Gamma^c$.  By the constrain \eqref{eq:prodcon}, we immediately obtain
\[
\mbox{quot}:=\frac{\prod_{k\in \Gamma} u_k}{\prod_{k\in \Gamma^c} u_k} = \frac{r(N-1)}{|x(0)|^2}  =u_N.
\]

{\bf Case 2:} There is an index $k_j$ such that $1/\beta_{k_j}$ is a zero of $X(z)$. Using Lemma \ref{le:key2}, we have that $\max\xkh{ X_m({-u_{k_j}}), X_m({-u_{k_j}^{-1}})} \le  u_{\max}^{-N} \le  1/4$ due to $u_{\max}\ge 3$.
On the other hand, if $\max\xkh{ X_m({-u_{k}}), X_m({-u_{k}^{-1}})} \le  1/4$ then $X(1/\beta_k)=0$. Indeed, if $X(1/\beta_k)\neq 0$, then
we can use (\ref{eq:Xmcase1})  to obtain that $1/4\geq \abs{X_m(1/\beta_k)} \geq \abs{X_m(\beta_k)}+3/4$. A contradiction.
This gives a rule to determine which index belongs to this case where $X(\beta_k)=X(1/\beta_k)=0$ (See line 7 in Algorithm \ref{alg2}). Furthermore, from the construction of $\vx$ given in \eqref{eq:zerosbeta22}, we know that the case where both $\beta_{k_j}$ and $1/\beta_{k_j}$ are zeros of $X(z)$ happens only when there is an index $k_l$ such that $u_{k_l}=u_{k_j}$,  and $k_j \in \Gamma, k_l\in \Gamma^c$
(or $k_l \in \Gamma, k_j \in \Gamma^c$).  Thus,  from \eqref{eq:orpart22}, it holds
\[
\prod_{k\in \Gamma  \setminus \dkh{k_j}} u_k = u_{N} \cdot \prod_{k\in \Gamma^c \setminus \dkh{k_l}} u_k.
\]
Thus,  if the original product partition problem $\left\{ u_1,\ldots,u_{N} \right\}$  has a solution then the following product partition problem $\left\{ u_1,\ldots,u_{N} \right\}\setminus \dkh{{u_{k_l},u_{k_j}}}$ also has a solution. We can then apply Algorithm $\mathcal A$ to find this solution recursively.

Therefore, for the case where the product partition problem $ u_1,\ldots,u_{N}$ has a solution,  Algorithm \ref{alg2} could find this solution by applying the algorithm $\mathcal{A}$  at most $\lceil N/2 \rceil$ times,  and we have
\begin{equation} \label{eq:panduan22}
\mbox{quot}:= \frac{\prod_{k\in \Gamma_1} u_k}{ \prod_{k\in \Gamma_2} u_k }=u_N,
\end{equation}
where the subsets $\Gamma_1$ and $\Gamma_2$ are given by Algorithm \ref{alg2}.

We next turn to the case where the Product Partition problem $u_1,\ldots,u_N \ge 2$ does not have a solution.  For this case, by the definition,  the   subsets $\Gamma_1$ and $\Gamma_2$ given by Algorithm \ref{alg2} must obey
 \[
 \mbox{quot}= \frac{\prod_{j\in \Gamma_1} u_j }{ \prod_{j\in \Gamma_2} u_j} \neq u_N.
 \]

In summary,  by running Algorithm \ref{alg2},  if $\mbox{quot}:= \frac{\prod_{k\in \Gamma_1} u_k}{ \prod_{k\in \Gamma_2} u_k }=u_N$,  then the original product partition problem has a solution; otherwise, it does not have a solution.

 Finally, note that the algorithm $\mathcal A$ runs at most $\lceil N/2 \rceil$ times in Algorithm \ref{alg2},   and   $m=\mbox{Poly}(N) \log u_{\max}$ steps are executed for each time. Thus, the total steps of Algorithm \ref{alg2} is
\begin{equation} \label{eq:totalstep22}
\mbox{Total}:= \lceil N/2 \rceil \cdot \xkh{\mbox{Poly}(N) \log u_{\max}+O(N) } = \mbox{Poly}(N) \log u_{\max}.
\end{equation}
Observe that the bit-length of the input of Product Partition $ u_1,\ldots,u_{N}$ is  $L:=N\log u_{\max}$.  Therefore, the total steps given in \eqref{eq:totalstep22} is a polynomial function of $L$. This means  Algorithm \ref{alg2}  could give an answer to Product Partition  in polynomial time,  which completes the proof of the claim.
\end{proof}

%\begin{remark}
%Theorem \ref{th:main} demonstrates that no algorithm could gives an estimate $\vx_m$ such that
%\[
%\norm{\vx_m-\vx} \le \epsilon
%\]
%within  $O(\log 1/\epsilon)$ steps, where $\vx$ is the true solution. Actually, using the same argument in Theorem \ref{th:main}, we could show that the Fourier phase retrieval problem can not be solved in $O(\mbox{poly}(N)\log 1/\epsilon)$ steps to reach $\epsilon$-precision. Here, $N$ is the dimension of the true signal $\vx$.
%\end{remark}

%\begin{remark}
%Note that find the solution of  Product Partition problem is NP-complete in the strong sense. Using the same argument, we could obtain that there is no sub-linear algorithm with $O(1/k)$ convergence rate for solving the Fourier phase retrieval. Indeed, if there is a sub-linear algorithm, then it requires only
%\[
%m=O(u_1^N)
%\]
%iterations to reach the error $ \norm{\vx_m-\vx} \le 3^{-2Nu_1}$. It means there exists a pseudo-polynomial for solving the product partition problem, which contradicts the fact that product partition problem is NP-complete in the strong sense.
%\end{remark}

\section{Discussions}
This paper presents a theoretical understanding of the algorithms for solving Fourier phase retrieval, showing that there is no linear algorithm to reconstruct a signal $\vx\in\C^N$ from its Fourier magnitude and the initial value $x(0)$. The result demonstrates that the fundamental gap between theory and algorithms for Fourier phase retrieval can not be reduced in general.

There are some interesting problems for future research. First, the arguments employed  in this paper does not preclude the existence of a sublinear algorithm for solving the Fourier phase retrieval problem. It is of practical interest to propose a sublinear algorithm. Secondly,  the theoretical understanding of algorithms for sparse Fourier phase retrieval is still limited, especially when the sparsity is greater than $O(N^{1/2})$. It is interesting to extend the result in this paper to sparse Fourier phase retrieval.

%\bibliography{bibfile}

\end{document}